\documentclass{article}
\usepackage{amsmath}
\usepackage{amsthm}
\usepackage{amssymb}
\usepackage{bbm}
\usepackage{amsfonts}
\usepackage{graphicx}
\usepackage{float}
\usepackage{enumerate}
\usepackage{color}
\usepackage{cleveref}
\usepackage{booktabs}
\usepackage{natbib}
\usepackage{url}
\usepackage{bm}
\usepackage[normalem]{ulem}

\newtheorem{theorem}{Theorem}[section]
\newtheorem{definition}[theorem]{Definition}

\newtheorem{remark}[theorem]{Remark}
\newcommand{\ed}{\mathrm{d}}

\definecolor{darkgreen}{RGB}{0,127,0}

\renewcommand{\P}{{\mathbb P}}
\newcommand{\Q}{{\mathbb Q}}
\newcommand{\E}{{\mathbb E}}

\newcommand{\R}{{\mathbb R}}

\newcommand{\ES}{\text{ES}}
\newcommand{\bs}{\boldsymbol}

\DeclareMathOperator{\ess}{ess}
\DeclareMathOperator{\GARCH}{GARCH}

\title{The ineffectiveness of coherent risk measures}
\author{John Armstrong \\ Dept. of Mathematics \\ King's College London \\ \tt{john.armstrong@kcl.ac.uk} \and Damiano Brigo\thanks{Corresponding author}\\ Dept. of Mathematics \\ Imperial College London \\ \tt{d.brigo@ic.ac.uk}}
\begin{document}

\maketitle

\begin{abstract}
We show that coherent risk measures are ineffective in curbing 
	the behaviour of investors with limited liability or excessive tail-risk seeking behaviour if the market admits
	statistical arbitrage opportunities which we term $\rho$-arbitrage
	for a risk measure $\rho$. We show how to determine analytically whether such $\rho$-arbitrage portfolios
	exist in complete markets and in the Markowitz model. We also consider realistic
numerical examples of incomplete markets and determine whether expected shortfall constraints
are ineffective in these markets. We find that the answer depends heavily upon the probability
model selected by the risk manager but that it is certainly possible for expected shortfall
constraints to be ineffective in realistic markets. Since value at risk constraints are weaker
than expected shortfall constraints, our results can be applied to value at risk. By contrast, we show that reasonable expected utility constraints  are effective in any arbitrage-free market.
\end{abstract}

\bigskip


\medskip

\section*{Introduction}

In \cite{armstrongbrigoutilityjbf} it was shown that neither value-at-risk constraints nor expected-shortfall constraints are sufficient to curb the behaviour of a risk-seeking trader
or risk-seeking institution in typical complete markets.

In that paper a risk-seeking trader was modelled as an investor who wishes to optimize an S-shaped utility curve. This is an increasing utility curve that is convex on the left and concave on the right. An investor with such a utility curve is conventionally risk-averse in profitable situations, but risk-seeking in loss-making situations. This idea is motivated by the theory of \cite{kahnemanAndTversky} who observed such behaviour empirically. It can also be justified theoretically by observing that traders have limited liability as they can lose no more than their job and their reputation. Similarly shareholders in banks have limited liability.

\cite{armstrongbrigoutilityjbf} shows that subject to mild technical assumptions, in complete markets, including the case of the Black-Scholes model, a trader with S-shaped utility operating under cost and expected shortfall constraints can achieve any desired expected utility, bounded only by the supremum of their utility function. Moreover it is shown that the trader will take a position with infinitely bad utility if that utility were to be measured with a conventional concave utility function. We will refer to such a utility function as the risk-manager's utility function. This suggests that, in this context, expected conventional utility could be a more effective risk measure for the risk manager than expected shortfall and VaR.

The most significant assumption of \cite{armstrongbrigoutilityjbf} is that the market is complete. This paper seeks to ask how the results change if one studies incomplete markets or other coherent risk measures (as introduced by \cite{artznerEtAl}).

We begin by identifying why expected shortfall is ineffective in Section \ref{section:theory}. We give a formal definition of what we mean for a risk measure to be ineffective in terms
of whether it successfully reduces the utility that can be achieved by the ``worst-case'' trader, namely a risk-neutral trader with limited liability. We show that
coherent risk measures, $\rho$, are ineffective if and only if
the market contains a specific type of portfolio associated with $\rho$ which we call a $\rho$-arbitrage portfolio. These are portfolios which give a potentially positive return for a non-positive price without incurring a positive risk as measured using $\rho$. We will see that if the market
admits a $\rho$-arbitrage, then the risk constraint will be ineffective against traders with a broad range of S-shaped utility functions, not just the worst-case example. Our results in this section apply to rather general markets, our key assumption is that the market is positive-homogeneous, implying that unlimited quantities of assets can be purchased at a given price.
Our analysis shows that the key fault of expected shortfall is that it is
a coherent risk measure, in particular it is positive-homogeneous. We conclude that if one wants to use risk-measures that are effective one should consider convex measures, as introduced by \cite{follmer1}.

After Section \ref{section:theory} the paper focuses on expected shortfall. We will write $\ES_p$ for the expected shortfall at confidence level $p$. Our theory tells us that to determine if $\ES_p$ is effective we must look for $\ES_p$-arbitarge portfolios. 

In Section \ref{section:analytic} we show how this characterization of ineffectiveness allows
us to compute analytically whether or not expected shortfall is effective at a given confidence level in certain simple markets. We first consider the case of a market in assets which follow a multi-variate normal distribution, as considered by \cite{markowitz}. We next consider the case of complete markets. Our results show that $\ES_p$ arbitrage is unlikely in the Markowitz model for low values of $p$ but inevitable for all $p$ in the Black--Scholes model.

In Section \ref{section:numerics} we show how our characterization of ineffectiveness can be
used in practice for realistic markets. Using the techniques of \cite{rockafellarUrayasev},
we are able to give a practical numerical method for determining if $\ES_p$ is effective.
We demonstrate this technique in practice by considering the market of European options 
with a fixed maturity on the S\&P 500.
We find that the values for which $\ES_p$ is effective depend heavily upon the probability model chosen. We find that for some ostensibly reasonable fat-tailed probability models calibrated to market data, $\ES_p$ is ineffective for even very low values of $p$. In particular we found this for a $\GARCH(1,1)$ model calibrated to historic data and a mixture model calibrated to option prices. Thus the ineffectiveness of $\ES_p$ constraints
found in \cite{armstrongbrigoutilityjbf} cannot be put down simply to the use of an idealised market model.

The paper ends with an appendix which collects together the proofs.

We finally note that we are not the only authors to have identified that positive-homogeneity as a problematic property for a risk-measure.  \cite{herdegen} have independently reached a similar conclusion by studying a concept which they term regulatory arbitrage, which is closely related to $\rho$-arbitrage.
They define regulatory arbitrage as occuring when the problem of maximising expected return subject to $\rho$ constraints becomes ill-posed.
The possibility of such ill-posed problems was first noticed by \cite{alexanderBaptista} in the case of VaR. See
\cite{herdegen} for a review of the subsequent literature. While the existing literature focuses on the problem of maximising an expected return, and hence on the behaviour of risk-neutral agents, we consider the further dangers posed by agents which are tail-risk-seeking. In addition we study the effectiveness of expected utility constraints against such agents, thereby demonstrating that the strategies taken by tail-risk-seekers under positively homogeneous constraints will yield arbitrarily low utilities for the risk-manager, who would find such trades unacceptable.

\section{Ineffective constraints and $\rho$-arbitrage}
\label{section:theory}

We will begin by giving general definitions of a financial market. We wish to give definitions that are broad enough to include incomplete markets where there may be a bid-ask spread or even an order book. Our treatment is based on that of \cite{pennanen2011,pennanen2012}.

A {\em market} consists of a probability space $(\Omega, {\cal F}, \P)$ and a function 
\[{\cal P}:L^0(\Omega,\R)\to{\R}\cup \{+\infty\}.
\]
Each random variable represents the payoff of an asset and ${\cal P}$ computes the price of an asset. Assets with price $\infty$ cannot be purchased. Because we have included $\infty$ in the range of ${\cal P}$, we can safely assume that ${\cal P}$ is defined on the whole of $L^0(\Omega,\R)$ and not just some subset. (Mathematically one can extend the definition to allow
markets where ${\cal P}$ may take the value $-\infty$ on liabilities so bad that the market is willing to pay arbitrarily large sums to anyone willing to take on this liability. However, in this paper we restrict our attention to markets without such liabilities.)

\begin{definition}
A market is {\em positive-homogeneous} if ${\cal P}(\lambda X)=\lambda {\cal P}(X)$ for $\lambda \geq 0$. A market is {\em coherent} if it is positive-homogeneous and
\begin{equation}
{\cal P}(X+Y) \leq {\cal P}(X) + {\cal P}(Y)
\label{eq:subadditivity} 
\end{equation}
\begin{equation}
{\cal P}(1) < \infty, \quad {\cal P}(-1) < \infty.
\label{eq:riskFreeAsset}
\end{equation}
\end{definition}

 Note that by requiring only positive-homogeneity, one allows for a bid-ask spread. Assuming that a market is positive-homogeneous is an idealisation; for example it implies that there are no quantity constraints or price impact. However, once one has assumed positive-homogeneity, the assumption of sub-additivity \eqref{eq:subadditivity} is rather innocuous as one should be able to replicate the payoff $X+Y$ by purchasing the assets $X$ and $Y$ separately once one assumes there are no quantity constraints. Assuming positive-homogeneity, equation \eqref{eq:riskFreeAsset} is the assumption that there is a risk-free asset.

We use the term coherent simply by analogy with so-called coherent risk-measures. We do not wish to imply that there is anything logically incoherent about markets which are not coherent, the word is merely intended to convey the uniformity arising from positive-homogeneity.

A {\em trading constraint} ${\cal A}$ is a subset of the set of random variables representing the assets that a trader is allowed to purchase.

Let $\tilde{u}(x):=x^+$. This can  be thought of as the utility function of an risk-neutral investor with limited liability.

\begin{definition}[Ineffective Constraint]
	A trading constraint ${\cal A}$ is {\em ineffective} if for any cost $C \in \R$
	\[
	\sup_{X \in {\cal A},\, {\cal P}(X)\leq C}{ E(\tilde{u}(X)) } = \infty.
	\]
\end{definition}

Note that we include negative costs in this definition. So under ineffective constraints even heavily indebted traders with utility $\tilde{u}$ would be able to achieve arbitrarily large utilities from their investments while at the same time clearing their debts.

Since any conventional utility function or any of the S-Shaped utility functions studied
by Kahneman and Tversky can be bounded above by some affine transformation of $\tilde{u}(x)$,
 the utility function $\tilde{u}$ represents a worst-case scenario for the risk manager. Thus a trading-constraint
is ineffective if it is possible for a trader's expected utility to be unperturbed by
the constraint in this worst case scenario, but the constraint may still have an effect on less aggressively risk-seeking
traders.

\begin{definition}[$\rho$-arbitrage]
	If $\rho$ is a function on the space of random variables, then a random variable is called a ${\rho}$-arbitrage if ${\cal P}(X)\leq 0$, $\rho(X)\leq 0$ and $X$ has a positive probability of taking a positive value.
\end{definition}

We will define a true arbitrage to be a random variable $X$ which has a positive probability of being positive,
is almost surely non-negative and has a non-positive price. Although many
authors prefer to insist that an arbitrage has a price of zero, allowing negative prices is more natural from the point
of view of convex analysis.
 If $\rho^c$ assigns the value $c> 0$ to any random variable which takes negative values with positive probability, then a true arbitrage is equivalent to a $\rho^c$-arbitrage. This justifies the name $\rho$-arbitrage (we remark that a variance-arbitrage or a standard-deviation-arbitrage will also be a true arbitrage).

Functions $\rho$ on $L^\infty(\Omega;\R)$ that are intended to measure risk have been studied extensively,
notably by \cite{artznerEtAl}. They gave a set of axioms that $\rho$ must obey for it to be called
a coherent risk-measure. We adapt their definition slightly to match our conventions for the domain and range of $\rho$.
\begin{definition}
A {\em coherent risk measure} $\rho:L^0(\Omega;\R)\to \R \cup \{\infty \}$ with
$L^\infty(\Omega,\R) \subseteq \rho^{-1}(\R)$ satisfies
\begin{enumerate}[(i)]
    \item Normalization: $\rho(0)=0$
	\item Montonicity: $\rho(X)\geq \rho(Y)$ if $X \leq Y$ almost surely.
	\item Sub-additivity: $\rho(X_1 + X_2 ) \leq \rho(X_1) + \rho(X_2)$.
	\item Translation invariance: $\rho( X + a )=\rho(X) - a$ for $a \in \R$.
	\item Positive homogeneity: $\rho( \lambda X )=\lambda \rho(X)$ for $\lambda \in \R^+$.	
\end{enumerate}	
\end{definition}

We may now state our main theoretical results which connect the effectiveness of a risk-measure $\rho$
to the existence of $\rho$-arbitrage. Note that the axiom of positive homogeneity plays
a crucial role in their proofs, which can be found in Appendix \ref{appendix:proofs}.

\begin{theorem}
\label{thm:rhoArbitrageAndUtility} (Arbitrarily good trader utilities can be otbained if there is a $\rho$-arbitrage).
Let $\rho$ be a coherent risk-measure. If a
coherent market  contains a $\rho$-arbitrage $X$ then
for any random variable $Y$ of finite expectation
\begin{align}
\lim_{\lambda\to \infty} \E( \tilde{u}(Y+ \lambda X)) &= \infty \label{eq:goodForRogue} \\
{\cal P}(Y+ \lambda X) &\leq {\cal P}(Y) \label{eq:priceInequality} \\
\rho(Y+ \lambda X) &\leq \rho(Y). \label{eq:rhoInequality}
\end{align}
If in addition $E(X^-)<\infty$, then for utility functions of the form
\begin{equation}\label{eq:powerutility}
u(x):=\begin{cases}
C_1 x^{a_1} & x \geq 0 \\
-C_2 (-x)^{a_2} & x \leq 0
\end{cases}
\end{equation}
where $C_1>0$, $C_2 \geq 0$ and $0 < a_2 < a_1 \leq 1$  we have
\begin{equation}
\lim_{\lambda\to \infty} \E( u(Y+ \lambda X)) = \infty.
\label{eq:goodForRogue2}
\end{equation}
\end{theorem}
\begin{theorem}
\label{thm:rhoArbitrageCharacterisation} (Equivalence between existence of $\rho$-arbitrage and ineffectiveness). 	
In a market containing a $\rho$-arbitrage, the constraint
\[
{\cal A}^{\rho, \alpha}:=\{Y \mid \rho(Y)\leq \alpha\}
\]
is ineffective for all $\alpha$. Conversely if ${\cal A}^{\rho, \alpha}$
is ineffective then the market admits a $\rho$-arbitrage.
\end{theorem}

The first result shows that a $\rho$-arbitrage can be exploited by
a trader to obtain arbitrarily good utilities $\tilde{u}$. The second gives a characterisation of effectiveness in terms of $\rho$-arbitrage. Our next result
shows how the same portfolios perform when measured with typical conventional concave increasing utility functions $u_R$, which might be thought of as the utility function of the risk-manager, the business overall or wider society.

\begin{theorem}
\label{thm:badForRiskManager} (Arbitrarily good trader utilities are ruled out by a classic  utility used as risk measure).
	Let $\rho$ be a coherent risk-measure. Let $u_R$ be any concave increasing utility function satisfying 
	\begin{equation}
	\lim_{\lambda \to \infty} \frac{u_R(-\lambda)}{\lambda} = -\infty
	\label{eq:riskAverseUtility},
	\end{equation}
	If $X$ 
	is a $\rho$-arbitrage and not a true arbitrage, and if both $\E(|X|)$ and $\E(u_R(
	-\beta Y))$ are finite for some $\beta>0$ then
	\begin{align}
	\lim_{\lambda\to \infty} \E( u_R(Y+ \lambda X)) &= -\infty. \label{eq:badForRiskManager}
	\end{align}
\end{theorem}

Even very mildly risk-averse utility functions will satisfy \eqref{eq:riskAverseUtility} for example the function defined by
\[
u_{R,\eta}(x) = \begin{cases}
-(-x)^\eta & \text{when } x \leq 0 \\
0 & \text{otherwise}
\end{cases}
\]
satisfies \eqref{eq:riskAverseUtility} for any $\eta >1$. Thus for any such $u_R$, $\rho$-arbitrage opportunities give unbounded upward potential for the utility of a rogue investor and unbounded downward potential for the utility of a risk manager with utility $u_R$.

In a similar vein, the next theorem shows that utility based risk 
constraints will typically be effective in finite-dimensional linear markets.
\begin{definition}
A {\em finite-dimensional linear market} is a market where ${\cal P}^{-1}(\R)$
is a vector subspace of $L^0(\Omega,\P,R)$ and where ${\cal P}$ is a linear
functional on ${\cal P}^{-1}(\R)$.
\end{definition}
\begin{theorem}
\label{thm:utilityConstraintsEffective} (Classic utilities are ineffective as risk measures if and only if the market is arbitrageable in the classic sense).
Let $u_R$ be a concave increasing utility function satisfying 
\eqref{eq:riskAverseUtility} and $L \in \R$, then for a finite-dimensional
linear market with ${\cal P}^{-1}(\R)\subseteq L^1(\Omega, \P, \R)$, the set
\[
{\cal A}:=\{ Y \mid \E(u_R(Y)) \geq L \}
\]
is ineffective if and only if the market contains a true arbitrage. 
\end{theorem}

Theorem \ref{thm:rhoArbitrageCharacterisation} tells us that we can detect whether a given coherent risk measure $\rho$ leads to ineffective risk-constraints in a given coherent market ${\cal M}$ by solving the convex optimization problem
\[
\begin{array}{lcl}
\underset{X \in L^0(\Omega; \R)}{\text{minimize}} & & \ess \inf -X \\
\text{subject to} & & {\cal P}(X)\leq 0 \\
\text{and} & & \rho(X)\leq 0. \\
\end{array}
\]
Our assumptions on the coherence of the market and of $\rho$ ensure that constraints
are indeed convex. The minimum achieved will be negative (indeed it will then equal $-\infty$) if and only if $\rho$ is ineffective. Since this is a convex optimization problem it is relatively straightforward to solve in practice. We will use this method to find a number of
markets which contain a $\rho$-arbitrage in the later sections of this paper.

\section{Analytic results}
\label{section:analytic}

We will write $\ES_p$ for the coherent risk measure given by expected shortfall at confidence level $p$ \citep{acerbitasche}. This is defined by
\[
\ES_p(X) = \frac{1}{p} \int_0^p \text{VaR}_p(X) \ed p
\]
where value at risk at confidence level $p$, $\text{VaR}_p$, is defined in turn by
\[
\text{VaR}_p(X) = -\inf \{ x \in \R: F_X(x)>p \}
\]
where $F_X$ is the cumulative distribution function of $X$.

In this section we consider the question of when $\ES_p$-arbitrage opportunities
exist in some simple markets where we can find analytical results. We consider the contrasting
cases of a highly incomplete and a complete market.

In section \ref{section:markowitz} we will consider the markets of normally distributed assets as considered by \cite{markowitz}, this is a highly incomplete market. In section \ref{section:complete} we will consider complete markets. 
We will find that in the highly incomplete market of normally distributed assets $\ES_{0.01}$-arbitrage is unrealistic. Whereas in a typical complete market such as the Black-Scholes model  $\ES_p$-arbitrage should be expected for all $p$.

\subsection{Normally distributed assets}
\label{section:markowitz}

We suppose that we wish to invest in the market of the Markowitz model. We suppose
there are $N$ assets $X_1$, $X_2$, \ldots, $X_N$ whose payoffs follow a multivariate normal distribution with mean vector $\bs \mu$ and covariance matrix $\bs \Sigma$. A portfolio represented by the vector $\bs \alpha$ consists of $\alpha_j$ units of stock $j$. The expected return of this portfolio is ${\bs \mu}^\top {\bs \alpha}$ and the variance is ${\bs \alpha}^T {\bs \Sigma} {\bs \alpha}$. The cost of portfolio ${\bs \alpha}$ is assumed to be ${\bs c}^\top {\bs \alpha}$ for some vector ${\bs c}$. So ${\cal P}$ is given by:
\[
{\cal P}(X)=\begin{cases}
{\bs c}^\top {\bs \alpha} & \text{when } X = \sum_{j=1}^N \alpha_j X_j \\
\infty & \text{otherwise}.
\end{cases}
\]
We will suppose that, up to scale there is only one risk-free portfolio.

This defines a coherent market which we will call a Markowitz market with risk free asset.

\begin{theorem}
Suppose $p<0.5$. Let $E(p)$ denote the expected shortfall of a standard normal random variable at confidence level $p$. Then a Markowitz market with risk free asset admits a $\ES_p$-arbitrage
if and only either
\[
g \geq E(p)
\]
or 
\[
1 + R_F < 0
\]
where $g$ is the gradient of the
capital allocation line and $R_F$ is the risk free return.
\end{theorem}
\begin{proof}
A normally distributed asset $X$ with mean $\mu$ and standard deviation $\sigma$ satisfies
\[
\ES_p(X) = \sigma E(p) - \mu
\]
So ${\bs \alpha}$ represents a $\ES_p$-arbitrage portfolio if and only if
\begin{equation}
\sqrt{{\bs \alpha}^\top {\bs \Sigma} {\bs \alpha}} E(p)  -{\bs \mu}^\top{\bs \alpha} \leq 0
\label{eq:markowitzRiskCondition}
\end{equation}
and
\begin{equation}
{\bs c}^\top{\bs \alpha} \leq 0.
\label{eq:markowitzPriceCondition}
\end{equation}
By the classification of Markowitz markets in \cite{armstrongMarkowitz} we may assume
without loss of generality that
\[
\Sigma = \left(\begin{array}{cc}
1_{N-1} & 0 \\
0 & 0 
\end{array}
\right), \quad {\bs c}^\top=(0,0,\ldots, 0, 1),\quad {\bs \mu}^\top=(g,0,0,\ldots,0,1+R_F) 
\]
where $1_{N-1}$ is the identity matrix of size $N-1$ and $g\geq 0$. So equations \eqref{eq:markowitzPriceCondition} and \eqref{eq:markowitzRiskCondition} become
\begin{equation}
(\sum_{j=1}^{N-1} \alpha_j^2)^{\frac{1}{2}} E(p) - g \alpha_1 - (1+R_F) \alpha_N \leq 0
\label{eq:markowitzToSolve}
\end{equation}
and
\[
\alpha_N \leq 0
\]
respectively. If $1+R_F<0$ we can always solve \eqref{eq:markowitzToSolve} simply by choosing a sufficiently small value for $\alpha_N$. If $1+R_F\geq0$ then any solution to \eqref{eq:markowitzToSolve} must satisfy
\[
(\sum_{j=1}^{N-1} \alpha_j^2)^{\frac{1}{2}} E(p) \leq g \alpha_1
\]
and hence 
\[
|\alpha_1| E(p) \leq g \alpha_1.
\]
The result now follows.
\end{proof}

So for an $\ES_{0.01}$-arbitrage to exist in a Markowitz market with positive interest rates, one would require $g > 2.665$. This is an unrealistically steep
capital allocation line for investments over a time period of a year or less.

\subsection{$\ES_p$-arbitrage portfolios in complete markets}
\label{section:complete}

We now consider the case of complete markets.

\begin{theorem}
\label{thm:completeMarkets}
Let ${\cal M}$ be a complete market given by an
atomless probability space $(\Omega, {\cal F}, \P)$ equipped with a measure $\Q$ equivalent to $\P$.
We suppose that any $X \in L^\infty(\Omega; \R)$ can be purchased at the price
\[
{\cal P}(X) = e^{-rT} \E_\Q( X)
\]
where $T$ is the time horizon of the investment and $r$ is the risk-free rate.
This market admits an $\ES_p$-arbitrage if and only if
\[
\P\left( \frac{\ed \Q}{\ed \P} \geq \frac{1}{p} \right) > 0
\]
\end{theorem}
\begin{remark}
It already follows from
the results of \cite{armstrongbrigoutilityjbf} that expected shortfall is ineffective
for any confidence level $p$ in complete markets where the Radon-Nikodym derivative
\[
\frac{\ed \Q}{\ed \P}
\]
is essentially unbounded. Thus the new result in Theorem \ref{thm:completeMarkets}
is the proof of the converse. As was shown in \cite{armstrongbrigoutilityjbf, armstrongBrigoRiskMagazine}, in complete markets such as the Black--Scholes model
with non-zero market price of risk, one should expect $\frac{\ed \Q}{\ed \P}$
to be essentially unbounded and hence for expected shortfall to be ineffective at
all confidence levels.
\end{remark}

\section{Numerical results}
\label{section:numerics}

In this section we will see how one can detect whether $\ES_p$ exists in
a realistic market numerically. In Section \ref{section:numericsTheory} we will outline a
general approach to detecting $\ES_p$ arbitrage. In Section \ref{section:numericsExample}
we will apply this to the specific case of options on the S\&P 500.

\subsection{Detecting $\ES_p$ arbitrage numerically}
\label{section:numericsTheory}
	
Let us begin by introducing some notation. We will assume that there are $N_I$ available instruments that one can invest in at time $0$. The price of instrument $i$ is $p_i$ and the payoff at time $T$ is the random variable $f_i(\omega)$. We write ${\bf p}$ for the vector of the prices of each instrument, and ${\bf f}(\omega)$ for the random vector containing all the payoffs. The investor chooses a portfolio containing $x_i$ units of instrument $i$. We write ${\bf x}$ for the vector with components $x_i$. To model a bid ask spread, we require that each $x_i \geq 0$ and model shorting a security as purchasing positive quantities of an asset with a negative price.

To find $\ES_p$-arbitrage portfolios, we seek portfolios of negative cost and negative expected shortfall. Thus we will consider the convex optimization problem:
\begin{equation}
\begin{aligned}
\underset{{\bf x}}{\text{minimize}} & & ES_p( \bf{x} ) & & \\
\text{subject to} & & & & \\
\text{cost constraint} & & {\bf p}\cdot {\bf x} \leq 0, & & \\
\text{quantity constraints} & & 0 \leq x_i \leq 1 & & (1 \leq i \leq N_I). \\
\end{aligned}
\label{eq:esArbitrageProblem}
\end{equation}
If the minimizing portfolio has strictly negative expected shortfall then it must be
a $\ES_p$-arbitrage portfolio. Note that we impose an upper bound constraint on each $x_i$ in order to ensure that the optimization problem always has a finite solution. Since one can always rescale an $\ES_p$-arbitrage portfolio, this additional upper bound constraint is harmless.

To solve this convex optimization problem in practice, we may use the techniques of \cite{rockafellarUrayasev}. Theorem 1 of that paper proves that
\[
\ES_p(x) = \min_\alpha F_p({\bf x},\alpha)
\]
where we define
\[
F_p({\bf x}, \alpha) = \alpha +  \frac{1}{p} \E ((-{\bf f}(\omega)\cdot {\bf x}-\alpha)^+ ).
\]

We next choose a quadrature rule for expectations with $N_Q$ evaluation points ${\Omega}_i$ and weights $w_i$ so that we can make the approximation
\begin{equation}
\E( g ) \approx \sum_{i=1}^{N_Q}  w_i \, g(\Omega_i)
\label{eq:quadRule}
\end{equation}
for suitably well-behaved random variables $g$. We may then approximate the expected shortfall as
\[
\ES_p(x) \approx \min_\alpha \tilde{F}_p({\bf x},\alpha)
\]
where
\[
\tilde{F}_p({\bf x}, \alpha) =  \alpha + \frac{1}{p} \sum_{i=1}^{N_Q} w_i \, (-{\bf f}(\Omega_i)\cdot {\bf x}-\alpha)^+.
\]
Following \cite{rockafellarUrayasev}, if we introduce auxiliary variables $u_i$ to replace the terms $(-{\bf f}(\Omega_i)\cdot {\bf x}-\alpha)^+$ we may approximate \eqref{eq:esArbitrageProblem} with the linear programming problem:
\begin{equation}
\begin{aligned}
\underset{\alpha,{\bf x}, {\bf u}}{\text{minimize}} & & \alpha + \frac{1}{p} \sum_{i=1}^{\tilde{N}_Q} w_i \, u_i & & \\
\text{subject to} & & & & \\
\text{cost constraint} & & {\bf p}\cdot {\bf x} \leq C, & & \\
\text{quantity constraints} & & 0 \leq x_i \leq 1 & & (1 \leq i \leq N_I) \\
\text{auxiliary constraints} & & u_i \geq 0 & & \\
 & & u_i \geq -{\bf f}(\Omega_i)\cdot {\bf x}-\alpha & & (1 \leq i \leq \tilde{N}_Q). \\
\end{aligned}
\label{eq:esArbitrageProblemApprox}
\end{equation}

The simplest choice of quadrature rule is Monte Carlo. We simply simulate $N_Q$ sample points $\Omega_i$ and give each sample point equal weight.

However, in the one dimensional case where the underlying is a single stock price $S_T$, if we know how to compute the integrals
\begin{equation}
\int_{p_1}^{p_2} p(S_T) \, \ed S_T, \qquad \int_{p_1}^{p_2} S_T p(S_T) \, \ed S_T, 
\label{eq:typicalIntegral}
\end{equation}
analytically, a better choice of quadrature rule can be obtained by first choosing integration points $S_1< S_2< \ldots < S_{N_Q}$ and then selecting the weights $w_i$ such that the quadrature rule is exact for payoff functions which are continuous and linear except at the points $S_2, S_3, \ldots, S_{N_Q-1}$. If the points $S_i$ include all the strike prices of European puts and calls available in the market, then all possible portfolio payoffs for European option portfolios will be of this form.

As we have seen, once we have chosen our quadrature rule we can find out if a $\ES_p$-arbitrage portfolio exists by solving the optimization problem \eqref{eq:esArbitrageProblemApprox}. We can then use the method of bisection to find the lowest $p$ for which $\ES_p$ arbitrage portfolios exist.

\subsection{$\ES_p$-arbitrage opportunities on the S\&P 500}
\label{section:numericsExample}

We apply the theory of Section \ref{section:numericsTheory} to the market of European options on the S\&P 500. We
consider buy and hold strategies in exchange traded European options on this index. We only consider portfolios where all the options expire on the same maturity date and use this maturity date as the time horizon in our computation of expected shortfall. 

For every day in the week commencing 10 Feb 2014, we obtained bid and ask prices for all the exchange traded options on the S\&P 500 with maturity 22nd March 2014 \citep{bloomberg}.
This data determines our pricing function ${\cal P}$ for a portfolio of options.

We must then choose a probability model for the S\&P 500 Index value on the maturity date. We may then view the option payoffs as random variables in this probability model, and this will describe
the market in full. The idea of studying this market is taken from \cite{udomsak}.

The choice of probability model for the index value is subjective. We considered
the following possibilities:
\begin{enumerate}[(i)]
	\item A $\GARCH(1,1)$ model for the log returns,
		  calibrated to the same historic return data. This was estimated
		  using the MATLAB functions \verb=garch= and \verb=estimate=. We then simulated $10^6$
		  returns to obtain a Monte Carlo quadrature rule for this model.
	\item We calibrated a $\Q$-measure probability model ${\cal M}^\Q$ given by a mixture of two normal to fit the market volatility smile for the options and assumed this was also
	the $\P$-measure model. Mixture dynamical models have been used under the pricing measure $\Q$ for smile modelling, see for example \cite{brigomix,alexander2004normal}.
	
Mixture models have also been used under the measure P  for portfolio allocation, see for example the work by Roncalli and co-authors \cite{ronc1, ronc2}, or for inclusion of liquidity risk in risk measures via random holding period, see \cite{brigonordio}. Our choice of a mixture model as both $\P$ and $\Q$ measure is not intended to be a realistic approach to choosing a $\P$-measure model, simply an attempt to find a statistical model that is close to market prices to discover whether $\ES_p$-arbitrage persists even when the $\P$ and $\Q$ measures are very close. The fit of the calibrated model to the market volatility smile
	is shown in figure \ref{fig:mixturesmile}. Because the integrals \eqref{eq:typicalIntegral} can be computed analytically for this model we were able to test whether $\ES_p$ arbitrage exists without using Monte Carlo quadrature.
\end{enumerate}

\begin{figure}[htbp!]
	\centering
	\includegraphics[width=0.7\linewidth]{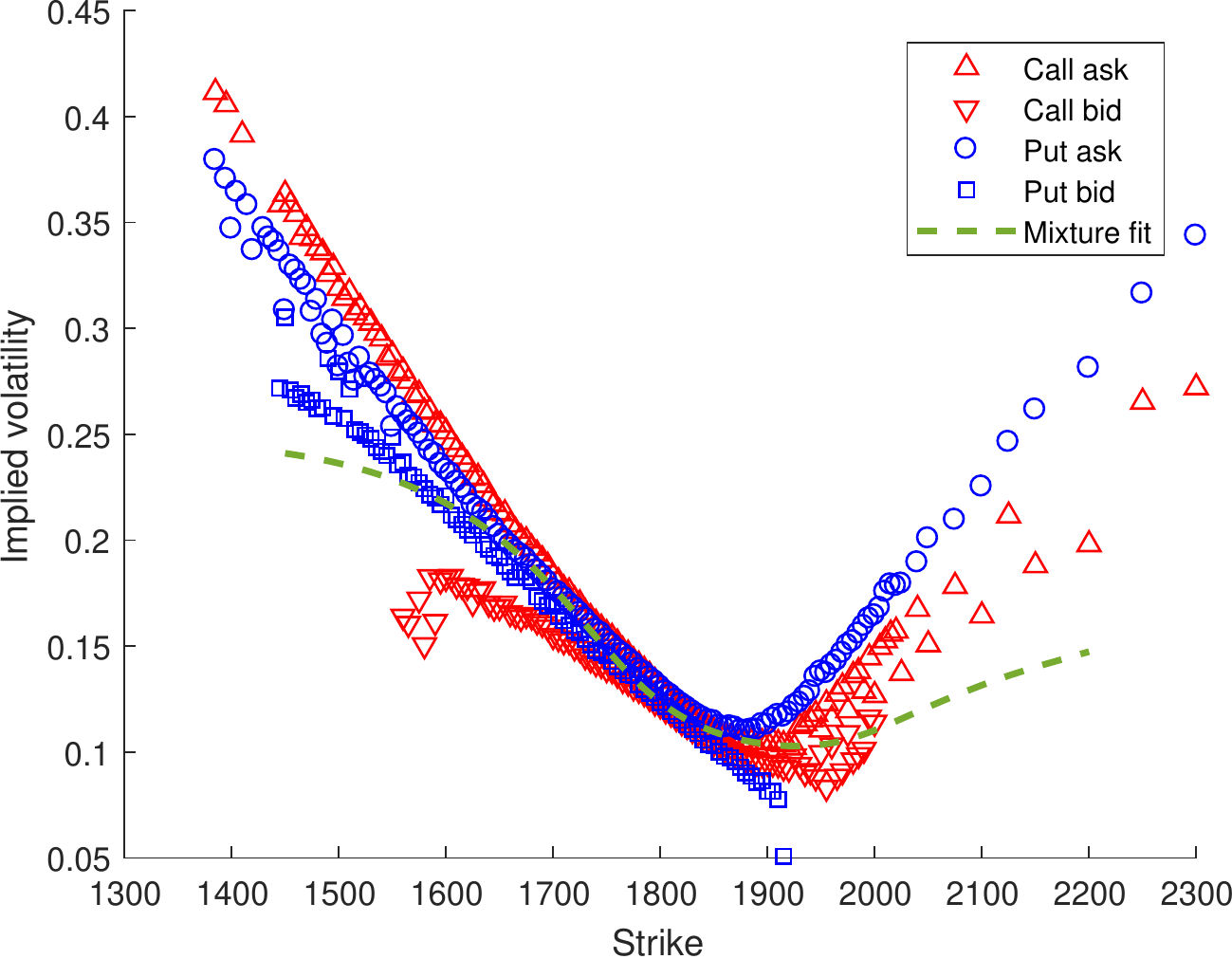}
	\caption{The volatility smile for 12 Feb showing the fit of a mixture of two log normal distributions, ${\cal M}_\Q$ calibrated to the market data.}
	\label{fig:mixturesmile}
\end{figure}

The results are shown in Table \ref{table:garchResults}.
\begin{table}[htbp]
\begin{centering}
\begin{tabular}{lrrr} \toprule
Date &  $\GARCH(1,1)$ run 1  &  $\GARCH(1,1)$ run 2  & Mixture  \\ \midrule
10 Feb & $<0.01\%$ & $0.19\%$ & $<0.01\%$  \\
11 Feb & $0.29\%$ & $<0.01\%$ & $<0.01\%$  \\
12 Feb & $0.33\%$ & $0.39\%$ & $<0.01\%$  \\
13 Feb & $<0.01\%$ & $<0.01\%$ & $<0.01\%$ \\
14 Feb & $0.26\%$ & $<0.30\%$ & $<0.01\%$ \\ \bottomrule
\end{tabular}
\caption{Minimum $p$ for which there exists a $\text{ES}_p$-arbitrage portfolio in exchange traded S\&P 500 options on the given data. For the $\GARCH(1,1)$ model, we performed two runs of the calculation in order to estimate the error produced by the use of Monte Carlo quadrature.}
\label{table:garchResults}
\end{centering}
\end{table}

Our results show that $\ES_p$-arbitrage opportunities can exist in real markets for low values of $p$ and with reasonable choices of $\P$-measure model.

\section{Conclusions}

We have shown that the ineffectiveness of expected shortfall as a means of
controlling the behaviour of tail-risk-seeking investors stems from its positive-homogeneity. Whether a positive-homogeneous risk constraint is effective or not depends upon whether or not the market contains a $\rho$-arbitrage. This is undesirable as risk constraints are typically set without reference to market conditions.

In the idealisation of the Black-Scholes model, expected-shortfall arbitrage opportunities exist for any confidence level $p$. We have shown how one can determine efficiently whether such an arbitrage exists and have used this to show that expected-shortfall arbitrage for small values of $p$ may exist in more realistic market models, including incomplete markets featuring transaction costs.

Positively-homogeneous risk measures have a natural attraction in that a regulator can use them to impose risk constraints that are proportionate to the size of investments and so, superficially, appear to treat large and small institutions ``fairly''. For example, in \cite{artznerEtAl}, the axiom of sub-additivity (which follows from convexity and positive homogeneity) is justified in part by the observation: ``If a firm were forced to meet a requirement of extra capital which did not satisfy this property [sub-additivity], the firm might be motivated to break up into two separately incorporated affiliates, a matter of concern for the regulator.'' From a post-crisis perspective, this argument seems to have lost
some of its persuasive force.

We would argue that a larger institution should be able to manage risk more effectively, making it reasonable for a regulator to insist that as institutions scale they should abide by
increasingly stringent risk constraints. We believe our results show that not only is it reasonable for a regulator to insist upon this, but it is essential if the regulator's constraints are intended to be effective irrespective of market conditions.


\bibliography{sshaped}
\bibliographystyle{apalike}

\appendix

\section{Proofs}
\label{appendix:proofs}

\begin{proof}
[Proof of Theorem \ref{thm:rhoArbitrageAndUtility}]
Let $X$ and $Y$ be as in the statement of the theorem.	
	
Using the subadditivity of $\rho$, followed by its positive homogeneity, followed by the definition of $\rho$-arbitrage we find:
\begin{equation*}
\rho(Y+ \lambda X) \leq \rho(Y) + \rho( \lambda X) = \rho(Y) + \lambda \rho(  X) \leq \rho(Y).
\end{equation*}	
This proves \eqref{eq:rhoInequality}. The same argument applied to ${\cal P}$ proves \eqref{eq:priceInequality}.

For $\lambda>0$ we have
\begin{align}
E(u(\lambda X))&=E(u(\lambda X^+)) + E(u(-\lambda X^-)) \nonumber \\ 
&= \lambda^{a_1} E(u(X^+)) + \lambda^{a_2}(u(-X^-)) \nonumber \\
&=  \lambda^{a_1} ( E(u(X^+))
+ \lambda^{a_2-a_1} E(u(-X^-)) ) \nonumber \\
&\geq  \lambda^{a_1} ( E(u(X^+))
+ \lambda^{a_2-a_1} (C_2(a_2-1)-C_2 a_2 E(X^-)) ).
\label{eq:scalingBound}
\end{align}

The last bound arises by differentiating $u$ at $-1$ and using the concavity of $u$ on the left.

Define 
\[
\ell(x) = \begin{cases}
C_1 x^{a_1} - C_1 - C_2 & x \geq 1 \\
-C_2 (-x)^{a_2} & x \leq -1 \\
-C_2 & -1 \leq x \leq 1
\end{cases}
\]
so that $\ell(x)$ is continuous and increasing and satisfies the inequalities
\begin{equation}
u(x)-C_1-C_2 \leq \ell(x) \leq u(x).
\label{eq:ellInequalities}
\end{equation}
The function $\ell$ is differentiable except at $-1$ and $1$. At points where the derivative exists
\[
0 \leq \ell^\prime(x) \leq \max\{ C_1 a_1, C_2 a_2 \}=:\tilde{C}.
\]
Hence for all $y\geq 0$
\[
\ell(x-y)=\ell(x) - \int_{x-y}^x \ell^\prime(u) \, \ed u \geq \ell(x) - \tilde{C} y. 
\]
We deduce from \eqref{eq:ellInequalities} that for all $y\geq 0$
\[
u(x-y) \geq u(x) - \tilde{C} y - C_1 - C_2. 
\]
Hence
\begin{align*}
E(u(\lambda X + Y))
&\geq E(u(\lambda X - Y^-)) \\
&\geq E(u(\lambda X ) - \tilde{C} Y^-
- C_1 - C_2 ) \\
&\geq \lambda^{a_1} ( E(u(X^+))
+ \lambda^{a_2-a_1} (C_2(a_2-1)-C_2 a_2 E(|X^-|))  - \tilde{C} E(Y^-) - C_1 - C_2
\end{align*}
by \eqref{eq:scalingBound}.

There is a positive probability that $X>0$, so $E(u(X^+))>0$. Hence if either $E(|X^-|)<\infty$ or $C_2=0$ we find
\[
\lim_{\lambda \to \infty} E(u(\lambda X + Y))= \infty.
\]
This establishes \eqref{eq:goodForRogue}
and \eqref{eq:goodForRogue2}.

\end{proof}
\begin{proof}
[Proof of Theorem \ref{thm:rhoArbitrageCharacterisation}]
Using the notation of the previous proof, to see that the constraint ${\cal A}^{\rho, \alpha}$ is ineffective, simply take $Y=\alpha$.
Then by \eqref{eq:goodForRogue}, $Y+\lambda X$ will give a solution of arbitrarily high $\tilde{u}$ utility which lies in ${\cal A}^{\rho, \alpha}$ by \eqref{eq:rhoInequality} and which has a cost of less than ${\cal P}(\alpha)$ by \eqref{eq:priceInequality}.

Let us assume that ${\cal A}^{\rho, \alpha}$ is ineffective. This implies that for all $M$
we can find $X_M \in {\cal A}^{\rho, \alpha}$ with ${\cal P}(X_M) \leq -\alpha$ and $\E(X_M^+) \geq M$.
We see that $\rho(X_M + \alpha) \leq 0$ and ${\cal P}(X_M+\alpha) \leq 0$.
If we take $M=-\alpha+1$ then $E(X_M^+) \geq -\alpha + 1$ so $X_M^+$ is greater
than or equal to $-\alpha+1$ with positive probability, and hence so is $X_M$. Therefore
$X_M+\alpha$ is greater than $1$ with positive probability. We conclude that $X_M + \alpha$ is a $\rho$-arbitrage.
\end{proof}

\begin{proof}
[Proof of Theorem \ref{thm:badForRiskManager}]
Let $\alpha=\frac{\beta}{1+\beta}$, so that $0 < \alpha < 1$. Since $u_R$ is concave
\begin{align*}
u_R\left( \alpha \lambda X \right) &= u_R\left( \alpha (\lambda X + Y) + (1-\alpha) \left(-\frac{\alpha}{1-\alpha} Y \right) \right) \\
&\geq \alpha u_R( \lambda X + Y) + (1-\alpha) u_R\left(-\frac{\alpha}{1-\alpha} Y\right) \\
&= \alpha u_R( \lambda X + Y) + (1-\alpha) u_R\left(-\beta Y\right)
\end{align*}
Rearranging we find
\[
u_R( \lambda X + Y) \leq \frac{1}{\alpha} u_R\left( \alpha \lambda X \right) - \frac{1}{\beta} u_R(-\beta Y).
\]
So, by our assumption that $\E(u_R(-\beta Y))$ is finite, it suffices to prove that
\begin{equation}
\lim_{\lambda \to \infty} \E(u_R\left( \alpha \lambda X \right)) = -\infty.
\label{eq:utilityWithoutY}
\end{equation}

Since $X$ is not a true arbitrage, we may choose $\epsilon>0$ such that
$\P(X\leq-\epsilon)=p>0$.
\begin{align*}
\E(u_R(\lambda X))&=p \, \E(u_R(\lambda X) \mid X \leq -\epsilon) + (1-p) \, \E(u_R(\lambda X) \mid X > -\epsilon) \\
&\leq p \, u_R(-\lambda \, \epsilon) + (1-p) \, \E(u_R(\lambda X) \mid X > -\epsilon)
\end{align*}
since $u_R$ is an increasing function. Since $u_R$ is concave, it is bounded from above by a linear function, so $u(x)\leq a\, x + b$ for some $a, b \in \R$. Hence
\begin{align*}
\E(u_R(\lambda X))
&\leq p \, u_R(-\lambda \, \epsilon) + (1-p) \, \E(a \lambda X + b \mid X > -\epsilon) \\
&\leq p \, u_R(-\lambda \, \epsilon) + (1-p) (a\,\lambda \, \E(  X  \mid X > -\epsilon ) + b) \\
&= \lambda p \left[ \frac{u_R(-\lambda \, \epsilon)}{\lambda} + \frac{1-p}{p} \left(a \, \E(  X  \mid X > -\epsilon ) + \frac{b}{\lambda} \right) \right] \\
\end{align*}
By \eqref{eq:riskAverseUtility} the term in square brackets will be negative for sufficiently large $\lambda$. Equation \eqref{eq:utilityWithoutY} and hence equation \eqref{eq:badForRiskManager} follow.
\end{proof}

\begin{proof}
[Proof of Theorem \ref{thm:utilityConstraintsEffective}]

By making an affine transformation of $u_R$ if necessary, we may
assume $u_R(x)\leq x$ for all $x \in \R$ and $u_R(x)\leq 0$ for all $x \leq 0$.

Suppose that ${\cal A}$ is ineffective.
Let $X_1, \dots, X_N$ be a basis for ${\cal P}^{-1}(\R)$.
Write ${\bm X}=(X_1,\dots X_N)$. Define the set
\[
{\cal C}:=\{ {\bm{\alpha}} \in \R^N \mid \E( u( {\bm \alpha} \cdot {\bm X} ) ) \geq \min \{0,L\}
\text{ and }
{\cal P}({\bm \alpha} \cdot {\bm X}) \leq 0 \}.
\]
The bound 
\[
E(u_R({\bm \alpha} \cdot {\bm X})) \leq E(\tilde{u}({\bm \alpha} \cdot {\bm X} )) \leq {\bm \alpha}^+ E({\bm X}^+)
\]
shows that since ${\cal A}$ is ineffective, ${\cal C}$ will be unbounded. Since ${\cal C}$
is also convex and contains the origin, it must contain some ray starting at the origin. 

Hence we may
find ${\bm \alpha}^*\neq0$ with $\lambda \, {\bm \alpha}^* \in {\cal C}$ for all $\lambda \in {\R}_{\geq 0}$.

Since
\[
\lambda {\cal P}( {\bm \alpha}^* \cdot {\bm X} ) \leq 0 \quad \forall \lambda \in {\R}_{\geq 0}
\]
we have that ${\cal P}( {\bm \alpha}^* \cdot {\bm X} )\leq 0$.

Suppose for a contradiction that there is a positive probability that ${\bm \alpha}^* \cdot {\bm X}$ is negative.
Then we may find $\delta, \epsilon>0$ such that $\P( {\bm \alpha}^* \cdot {\bm X} \leq -\delta ) \geq \epsilon$.
Hence using our bound $u_R(x)\leq x$ for all $x$ and $u_R(x)\leq 0$ for all $x\leq 0$  we have that for all $\lambda$
\begin{align*}
L \leq \E(u_R(\lambda {\bm \alpha}^* \cdot {\bm X})) &\leq \epsilon \, u(- \lambda \delta ) + \lambda \E( ({\bm \alpha}^* \cdot {\bm X})^+ ) \\
&= u_R(-\lambda \delta) \left[ \epsilon + \frac{\lambda \delta}{u_R(\lambda \delta)} \frac{\E( ({\bm \alpha}^* \cdot \bm X)^+)}{\delta} \right].
\end{align*}
Using \eqref{eq:riskAverseUtility} and the fact that $u_R$ is concave and increasing, we find that the right hand side tends to $-\infty$ as $\lambda \to \infty$, yielding
the desired contradiction.

Therefore ${\bm \alpha}^* \cdot {\bm X}$ is almost surely non-negative. Since ${\bm \alpha}\neq 0$ and the $X_i$ are
assumed to be linearly independent, we deduce that there is a positive probability that ${\bm \alpha}^* \cdot {\bm X}$
is positive. Hence ${\bm \alpha}^* \cdot {\bm X}$ is a true arbitrage.
\end{proof}

\begin{proof}
[Proof Theorem \ref{thm:completeMarkets}]
As described in \cite{armstrongbrigoutilityjbf}, since the market is atomless we can
find a uniformly distributed random variable $U$ such that the Radon-Nikodym derivative
\[
\frac{\ed \Q}{\ed \P}=q(U)
\]
for some positive decreasing function $q(u)$ of integral $1$ over $[0,1]$. (In the event that $\frac{\ed \Q}{\ed \P}$ has a continuous distribution we may simply take $U$ to the the image of $\frac{\ed \Q}{\ed \P}$ under its own cumulative distribution function, the atomless assumption allows us to find $U$ in the general case).

It follows from the theory of rearrangements
described in \cite{armstrongbrigoutilityjbf} that if an $\ES_p$-arbitrage $X$ exists, then there exists an $\ES_p$ arbitrage of the form $X=f(U)$ for some increasing function $f:[0,1]\to \R$.

If $X$ is an $\ES_p$ arbitrage then so is $X+\ES_p(X)$. To see this, first note that $\ES_p(X) \leq 0$ and so
\[
{\cal P}(X + \ES_p(X)) \leq {\cal P}(X) + {\cal P}(\ES_p(X)) = {\cal P}(X) + |\ES_p(X)|{\cal P}(-1) \leq {\cal P}(X) \leq 0.
\]
Second we have by the axioms of a coherent risk measure that
\[
\ES_p(X + \ES_p(X)) = \ES_p(X) + \ES_p(X)\ES_p(1) = 0.
\]
Next note that $\ES_p(Y) \geq (E(-Y))$ which shows that a non-positive random variable $Y$ can only have an expected shortfall of zero if it is constant and equal to zero. It follows that $X + \ES_p(X)$ is either constant or has a positive probability of being positive. But $X + \ES_p(X)$ is constant if and only if $X$ is constant, and $X$ cannot be constant since it is an $\ES_p$ arbitrage. Therefore $X + \ES_p(X)$ has a positive probability of being positive. We have now shown that 
$X + \ES_p(X)$ is an $\ES_p$-arbitrage as claimed.

It follows that we may restrict attention to looking for $\ES_p$ arbitrage of the form $X=f(U)$ with $f$ increasing and $\ES_p(f(U))=0$.

Given $\alpha\leq 0$ and $\beta \geq 0$, let $A_{p,\alpha,\beta}$ be the set of increasing functions $f:[0,1]\to \R$ which satisfy
$f([0,p])\subseteq[\alpha,\beta]$, $f([p,1])\subseteq[\beta, \infty)$ and which have $\ES_p(f(U))=0$. We will say that $f\in A_{p,\alpha,\beta}$
is a $\rho$-arbitrage if $f(U)$ is a $\rho$-arbitrage.
We have shown above that the market
contains an $\ES_p$-arbitrage if and only if some $A_{p,\alpha,\beta}$ contains
an $\ES_p$-arbitrage.

Define
an increasing function $g \in A_{p,\alpha,\beta}$ by
\[
g(u)=\begin{cases}
\alpha & \text{when } 0 \leq u \leq \tilde{p} \\
\beta & \text{when } \tilde{p} < u \leq 1 \\
\end{cases}
\]
where $\tilde{p} \in [0,p]$ is chosen to ensure that ${\text{ES}}_p(g(U)) = 0$. This requires
\[
\alpha \frac{\tilde{p}}{p} + \beta \frac{p-\tilde{p}}{p} = 0
\]
and hence
\begin{equation}
\tilde{p} = \frac{\beta p}{\beta - \alpha}.
\label{eq:pTilde}
\end{equation}
Since $f(u)\geq \beta$ for $u \in [p,1]$ we compute that for $f \in A_{p,\alpha,\beta}$
\begin{align}
\E_\Q( f(U) ) &= \int_0^p q(u) f(u) \ed u + \int_p^1 q(u) f(u) \ed u \nonumber \\
&\geq  \int_0^p q(u) f(u) \ed u + \int_p^1 q(u) \beta \ed u \nonumber \\
&=  \int_0^p q(u) f(u) \ed u + \int_p^1 q(u) g(u) \ed u \label{eq:priceEq1}
\end{align}
We may rewrite the first term on the right hand side as follows
\begin{align*}
\int_0^p q(u) f(u) \ed u
 &= \int_0^p q(u) g(u) \ed u + \int_0^{p} q(u)(f(u)-g(u)) \ed u \\
 &= \int_0^p q(u) g(u) \ed u \\
 &\quad +\int_{0}^{\tilde{p}} q(u)(f(u)-g(u)) \ed u + \int_{\tilde p}^p q(u)(f(u)-g(u)) \ed u
\end{align*}
But $f(u)\geq g(u)$ on $[0,\tilde{p}]$, $f(u)\leq g(u)$ on $(\tilde{p},p]$. Moreover $q$ is also a decreasing function. We deduce that
\begin{align*}
\int_0^p q(u) f(u) \ed u
&\geq \int_0^p q(u) g(u) \ed u \\
&\quad +\int_{0}^{\tilde{p}} q(\tilde{p})(f(u)-g(u)) \ed u + \int_{\tilde p}^p q(\tilde{p})(f(u)-g(u)) \ed u \\
&= \int_0^p q(u) g(u) \ed u  +\int_{0}^p q(\tilde{p})(f(u)-g(u)) \ed u \\
&= \int_0^p q(u) g(u) \ed u  + p q(\tilde{p}) (\ES_p(f) - \ES_p(g))  \\
&= \int_0^p q(u) g(u) \ed u. \\
\end{align*}
Using this we may obtain from
equation \eqref{eq:priceEq1} that
\[
\E_\Q( f(U) ) \geq \int_0^1 q(u) g(u) = \E_\Q(g(U)).
\]
Hence by definition of ${\cal P}$, we have shown that for $f \in A_{p,\alpha,\beta}$
\begin{equation}
{\cal P}( f(U) ) \geq {\cal P}(g(U)).
\label{eq:priceFGreaterThanPriceG}
\end{equation}
We deduce that if $A_{p,\alpha,\beta}$ contains any $\ES_p$-arbitrage then $g$ must be an $\ES_p$-arbitrage. This will be the case so long as $\beta>0$ and ${\cal P}(g(U))\leq 0$.

Let us now make the dependence of $g$ on $p$, $\alpha$ and $\beta$ explicit and write
\[
g_{p,\alpha,\beta}(u)=\begin{cases}
\alpha & \text{when } 0 \leq u \leq \tilde{p} \\
\beta & \text{when } \tilde{p} < u \leq 1 \\
\end{cases}
\]
We have
\begin{equation}
{\cal P}(g_{p,\alpha,\beta}(U))
= 
\int_0^{\tilde{p}} \alpha q(u) \ed u + \int_{\tilde{p}}^1 \beta q(u) \ed u
\label{eq:priceG}
\end{equation}
where $\tilde{p}$ is given in \eqref{eq:pTilde}. We note that
\[
\alpha = \left(1 - \frac{p}{\tilde p} \right) \beta.
\]
So we may rewrite \eqref{eq:priceG} as 
\begin{align}
{\cal P}(g_{p,\alpha,\beta}(U))
&= 
\int_0^{\tilde{p}} \left(1 - \frac{p}{\tilde{p}} \right) \beta q(u) \ed u + \int_{\tilde{p}}^1 \beta q(u) \ed u \nonumber \\
&= 
-\int_0^{\tilde{p}} \frac{p}{\tilde{p}} \beta q(u) \ed u + \int_0^1 \beta q(u) \ed u \nonumber \\
&= 
\beta\left(1 -\frac{p}{\tilde{p}} \int_0^{\tilde{p}}  q(u) \ed u \right)
\label{eq:priceGExpanded}
\end{align}
Viewed as a function of $\beta$, ${\cal P}(g_{p,\alpha,\beta}(U))$ must be decreasing by
\eqref{eq:priceFGreaterThanPriceG}. We also note that
\begin{equation}
\lim_{\alpha\to-\infty} {\cal P}(g_{p,\alpha,\beta}(U)) = \lim_{\tilde{p}\to 0+} {\cal P}(g_{p,\alpha,\beta}(U)) = \beta(1 - p \sup_{u\in(0,1)} (q(u)) )
\label{eq:gPriceLimit}
\end{equation}
where we have used \eqref{eq:priceGExpanded}, the fact $q$ is decreasing and the fundamental theorem of calculus. We deduce that if
\[
\sup_{u\in(0,1)} (q(u)) > \frac{1}{p}
\]
then $g_{p,\alpha,\beta}(U)$ will be an $\ES_p$ arbitrage for sufficiently small $\alpha$.
Suppose we have
\[
\sup_{u\in(0,1)} (q(u)) < \frac{1}{p}
\]
then $g_{p,\alpha,\beta}(U)$ will not be an $\ES_p$ arbitrage for any value of $\alpha$.
If we have equality
\[
\sup_{u\in(0,1)} (q(u)) = \frac{1}{p}
\]
then the limit in \eqref{eq:gPriceLimit} will be achieved for finite $\alpha$ if and only if $q$ attains its supremum on $(0,1]$. Hence $g_{p,\alpha,\beta}$ will be a ${\text ES}_p$-arbitrage for sufficiently small $\alpha$ if and only if this supremum is attained. The result follows.
\end{proof}

\end{document}